\documentclass{llncs}
\usepackage[utf8]{inputenc}
\usepackage[english]{babel}
\usepackage[normalem]{ulem}
\usepackage{amsmath}
\usepackage{amssymb}
\usepackage{amsfonts}
\usepackage{centernot}
\usepackage{fouriernc}
\usepackage{algorithm}
\usepackage{algorithmic}
\usepackage{graphicx}
\usepackage{hyperref}
\usepackage[utf8]{inputenc}
\usepackage[english]{babel}
 \usepackage{graphicx}
\usepackage{subcaption}
\usepackage[usenames, dvipsnames]{color}
\definecolor{red}{rgb}{1,0.2,0.2}

\newtheorem{assumption1}{Assumption}

\usepackage{amssymb}
\begin{document}
\title{Linear Temporal Logic Satisfaction in Adversarial Environments using Secure Control Barrier Certificates \thanks{This work was supported by the U.S. Army Research Office, the National Science Foundation, and the Office of Naval Research via Grants W911NF-16-1-0485, CNS-1656981, and N00014-17-S-B001 respectively.}}
%
%
\author{Bhaskar Ramasubramanian\inst{1} \and
Luyao Niu\inst{2} \and Andrew Clark\inst{2} \and \\
Linda Bushnell\inst{1} \and Radha Poovendran \inst{1}}
\authorrunning{B. Ramasubramanian et al.}
%
\institute{Department of Electrical and Computer Engineering, University of Washington, Seattle, WA 98195, USA.\\
\email{\{bhaskarr, lb2, rp3\}@uw.edu} \and
Department of Electrical and Computer Engineering, Worcester Polytechnic Institute, Worcester, MA 01609, USA.\\
\email{\{lniu, aclark\}@wpi.edu}}
\maketitle              
\begin{abstract}
This paper studies the satisfaction of a class of temporal properties for cyber-physical systems (CPSs) over a finite-time horizon in the presence of an adversary, in an environment described by discrete-time dynamics. 
The temporal logic specification is given in $safe-LTL_F$, a fragment of linear temporal logic over traces of finite length. 
The interaction of the CPS with the adversary is modeled as a two-player zero-sum discrete-time dynamic stochastic game with the CPS as defender. 
We formulate a dynamic programming based approach to determine a stationary defender policy that maximizes the probability of satisfaction of a $safe-LTL_F$ formula over a finite time-horizon under any stationary adversary policy. 
We introduce \emph{secure control barrier certificates} (S-CBCs), a generalization of barrier certificates and control barrier certificates that accounts for the presence of an adversary, and use S-CBCs to provide a lower bound on the above satisfaction probability. 
When the dynamics of the evolution of the system state has a specific underlying structure, we present a way to determine an S-CBC as a polynomial in the state variables using sum-of-squares optimization. 
An illustrative example demonstrates our approach. 

\end{abstract}

\section{Introduction}

Cyber-physical systems (CPSs) use computing devices and algorithms to inform the working of a physical system \cite{baheti2011cyber}. 
These systems are ubiquitous, and vary in size and scale from energy systems to medical devices. 
The wide-spread influence of CPSs such as power systems and automobiles makes their safe operation critical. 
Although distributed algorithms and systems allow for more efficient sharing of information among parts of the system and across geographies, they also make the CPS vulnerable to attacks by an adversary who might gain access to the distributed system via multiple entry points. 
Attacks on distributed CPSs have been reported across multiple application domains \cite{farwell2011stuxnet}, \cite{shoukry2013non}, \cite{slay2007lessons}, \cite{sullivan2017cyber}. 
In these cases, the damage to the CPS was caused by the actions of a stealthy, intelligent adversary. 
Thus, methods designed to only account for modeling and sensing errors may not meet performance requirements in adversarial scenarios. 
Therefore, it is important to develop ways to specify and verify properties that a CPS must satisfy that will allow us to provide guarantees on the operation of the system while accounting for the presence of an adversary.

In order to verify the behavior of a CPS against a rich set of temporal specifications, techniques from formal methods can be used \cite{baier2008principles}. 
Properties like safety, stability, and priority can be expressed as formulas in linear temporal logic (LTL) \cite{ding2014optimal}. 
These properties can then be verified using off-the-shelf model solvers \cite{cimatti1999nusmv}, \cite{kwiatkowska2011prism} that take these formulas as inputs. 
If the state space and the actions available to the agents are both finite and discrete, then the environment can be represented as a Markov decision process (MDP) \cite{puterman2014markov} or a stochastic game  \cite{bertsekas2015dynamic}. 
These representations have also been used as abstractions of continuous-state continuous action dynamical system models \cite{belta2017formal}, \cite{niu2018secure}. 
However, a significant shortcoming is that the computational complexity of abstracting the underlying system grows exponentially with the resolution of discretization desired \cite{chow1991optimal}, \cite{gordon1999approximate}. 

The method of barrier certificates (or barrier functions), which are functions of the states of the system was introduced in \cite{prajna2007framework}. 
Barrier functions provide a certificate that all trajectories of a system starting from a given initial set will not enter an unsafe region. 
The use of barrier functions does not require explicit computation of sets of reachable states, which is known to be undecidable for general dynamical systems \cite{lafferriere2001symbolic}, and moreover, it allows for the analysis of general nonlinear and stochastic dynamical systems. 
The authors of \cite{prajna2007framework} further showed that if the states and inputs to the system have a particular structure, computationally efficient methods can be used to construct a barrier certificate. 

Barrier certificates were used to determine probabilistic bounds on the satisfaction of an LTL formula by a discrete-time stochastic system in \cite{jagtap2018temporal}.  
A more recent work by the same authors \cite{jagtap2019formal} used control barrier certificates to synthesize a policy in order to maximize the probability of satisfaction of an LTL formula. 

Prior work that uses barrier certificates to study temporal logic satisfaction assumes a single agent, and does not study the case when the CPS is operating in an adversarial environment. 
To the best of our knowledge, this paper is the first to use barrier certificates to study temporal logic satisfaction for CPSs in adversarial environments. 
We introduce secure barrier certificates (S-CBCs), and use it to determine probabilistic bounds on the satisfaction of an LTL formula under any adversary policy. 
Further, definitions of barrier certificates and control barrier certificates in prior work can be recovered as special cases of S-CBCs.

\subsection{Contributions}

In this paper, we consider the setting when there is an adversary whose aim is to ensure that the LTL formula is not satisfied by the CPS (defender). 
The temporal logic specification is given in $safe-LTL_F$, a fragment of LTL over traces of finite length. 
We make the following contributions: 
\begin{itemize}
\item We model the interaction between the CPS and adversary as a two-player dynamic stochastic game with the CPS as defender. The two players take their actions simultaneously, and these jointly influence the system dynamics. 

\item We present a dynamic programming based approach to determine a stationary defender policy to maximize the probability of satisfaction of an LTL formula over a finite time-horizon under any stationary adversary policy. 

\item In order to determine a lower bound on the above satisfaction probability, we define a new entity called \emph{secure control barrier certificates (S-CBCs)}. S-CBCs generalize barrier certificates and control barrier certificates to account for the presence of an adversary. 

\item When the evolution of the state of the dynamic game can be expressed as polynomial functions of the states and inputs, we use sum-of-squares optimization to compute an S-CBC as a polynomial function of the states. 

%
\item We present an illustrative example demonstrating our approach. 
\end{itemize}

\subsection{Outline of Paper}

We summarize related work on control barrier certificates and temporal logic satisfaction in Section \ref{RelatedWork}. 
Section \ref{Preliminaries} gives an overview of temporal logic and game-theoretic concepts that will be used to derive our results. 
The problem that is the focus of this paper is formulated in Section \ref{ProblemFormulation}. 
Our solution approach is presented in Section \ref{SolutionApproach}, where we define a dynamic programming operator to synthesize a policy for the defender in order to maximize the probability of satisfaction of the LTL formula under any adversary policy. 
We define a notion of secure control barrier certificates to derive a lower bound on the satisfaction probability, and are able to explicitly compute an S-CBC under certain assumptions. 
Section \ref{Example} presents an illustrative example, and we conclude the paper in Section \ref{Conclusion}. 

\section{Related Work}\label{RelatedWork}

The method of barrier functions was introduced in \cite{prajna2007framework} to certify that all trajectories of a continuous-time system starting from a given initial set do not enter an unsafe region. 
Control barrier functions (CBFs) were used to provide guarantees on the safety of continuous-time nonlinear systems with affine inputs for an adaptive cruise control application in \cite{ames2016control}. 
The notion of input-to-state CBFs that ensured the safety of nonlinear systems under arbitrary input disturbances was introduced in \cite{kolathaya2018input},  
and safety was characterized in terms of the invariance of a set whose computation depended on the magnitude of the disturbance. 
The authors of \cite{steinhardt2012finite} relaxed the supermartingale condition that a barrier certificate had to satisfy in \cite{prajna2007framework} in order to provide finite-time guarantees on the safety of a system. 
The verification and control of a finite-time safety property for continuous-time stochastic systems using barrier functions was recently presented in \cite{santoyo2019verification}. 
Barrier certificates were used to verify LTL formulas for a deterministic, continuous-time nonlinear dynamical system in \cite{wongpiromsarn2015automata}. 
Time-varying CBFs were used to accomplish tasks specified in signal temporal logic in \cite{lindemann2019control}. 
A survey of the use of CBFs to design safety-critical controllers is presented in \cite{ames2019cbf}. 
The use of barrier certificates or CBFs in these works were all for continuous time dynamical systems and did not consider the effect of the actions of an adversarial player. 

Barrier certificates in the discrete-time setting were used to analyze the reachable belief space of a partially observable Markov decision process (POMDP) with applications to verifying the safety of POMDPs in \cite{ahmadi2019control}, and for privacy verification in POMDPs in \cite{ahmadi2018privacy}. 
The use of barrier certificates for the verification and synthesis of control policies for discrete-time stochastic systems to satisfy an LTL formula over a finite time horizon was presented in \cite{jagtap2018temporal} and \cite{jagtap2019formal}. 
These papers also assumed a single agent, and did not account for the presence of an adversary. 

The authors of \cite{niu2019tracking} used barrier functions to solve a reference tracking problem for a continuous-time linear system subject to possible false data injection attacks by an adversary, with additional constraints on the safety and reachability of the system. 
Probabilistic reachability over a finite time horizon for discrete-time stochastic hybrid systems was presented in \cite{abate2008probabilistic}. 
This was extended to a dynamic stochastic game setting when there were two competing agents in \cite{ding2013stochastic}, and to the problem of ensuring the safety of a system that was robust to errors in the probability distribution of a disturbance input in \cite{yang2018dynamic}.
These papers did not assume that a temporal specification had to be additionally satisfied. 
%

Determining a policy for an agent in order to maximize the probability of satisfying an LTL formula in an environment specified by an MDP was presented in \cite{ding2014optimal}. 
This setup was extended to the case when there were two agents- a defender and an adversary- who had competing objectives to ensure the satisfaction of the LTL formula in an environment specified as a stochastic game in \cite{niu2018secure}. 
These papers assume that the states of the system are completely observable, which might not be true in every situation. 
The satisfaction of an LTL formula in partially observable environments represented as POMDPs was studied in \cite{sharan2014finite} 
and the extension to partially observable stochastic games with two competing agents, each with its own observation of the state of the system, 
was formulated in \cite{bhaskar2019finite}. 

\section{Preliminaries}\label{Preliminaries}

In this section, we give a brief introduction to linear temporal logic and discrete-time dynamic stochastic games. 
Wherever appropriate, we consider a probability space $(\Omega, \mathcal{F}, \mathbb{P})$. 
We write $(X, \mathcal{B}(X))$ to denote the measurable space $X$ equipped with the Borel $\sigma-$algebra, and $\mathbb{R}_{\geq 0}$ to denote the set of non-negative real numbers. 

\subsection{Linear Temporal Logic}

Temporal logic frameworks enable the representation and reasoning about temporal information on propositional statements. 
\emph{Linear temporal logic (LTL)} is one such framework, where the progress of time is `linear'. 
An \emph{LTL formula} \cite{baier2008principles} is defined over a set of atomic propositions $\mathcal{AP}$, and can be written as: 
\begin{align*}
\phi:=\mathtt{T}|\sigma| \neg \phi | \phi \wedge \phi | \mathbf{X} \phi |\phi \mathbf{U} \phi, 
\end{align*}
where $\sigma \in \mathcal{AP}$, and $\mathbf{X}$ and $\mathbf{U}$ are temporal operators denoting the \emph{next} and \emph{until} operations. 
The semantics of LTL are defined over (infinite) words in $2^{\mathcal{AP}}$.

The syntax of linear temporal logic over finite traces, denoted $LTL_F$ \cite{de2013linear}, is the same as that of LTL. 
The semantics of $LTL_F$ is expressed in terms of \emph{finite-length words} in $2^{\mathcal{AP}}$. 
We denote a word in $LTL_F$ by $\eta$, write $|\eta|$ to denote the length of $\eta$, and $\eta_i$, $0<i<|\eta|$, to denote the proposition at the $i^{th}$ position of $\eta$. 
We write $(\eta, i) \models \phi$ when the $LTL_F$ formula $\phi$ is true at the $i^{th}$ position of $\eta$. 

\begin{definition}[$LTL_F$ Semantics]\label{LTLSemantics}
The semantics of $LTL_F$ can be recursively defined in the following way: 
\begin{enumerate}
\item $(\eta,i) \models \mathtt{T}$;
\item $(\eta,i) \models \sigma$ iff $\sigma \in \eta_i$; 
\item $(\eta,i) \models \neg \phi$ iff $(\eta,i) \not \models \phi$; 
\item $(\eta,i) \models \phi_1 \wedge \phi_2$ iff $(\eta,i) \models \phi_1$ and $(\eta,i) \models \phi_2$; 
\item $(\eta,i) \models \mathbf{X} \phi$ iff $i<|\eta|-1$ and $(\eta,i+1) \models \phi$; 
\item $(\eta,i) \models \phi_1 \mathbf{U} \phi_2$ iff $\exists j \in [i, |\eta|]$ such that $(\eta,j) \models \phi_2$ and for all $k \in [i, j), (\eta,k) \models \phi_1$. 
\end{enumerate}
Finally, we write $\eta \models \phi$ if and only if $(\eta, 0) \models \phi$. 
\end{definition}

Moreover, the logic admits derived formulas of the form: 
\emph{i)} $\phi_1 \vee \phi_2:=\neg(\neg \phi_1 \wedge \neg \phi_2)$; 
\emph{ii)} $\phi_1 \Rightarrow \phi_2:= \neg \phi_1 \vee \phi_2$; 
\emph{iii)} $\mathbf{F}\phi:=\mathtt{T} \mathbf{U} \phi  \text{  (eventually)}$; 
\emph{iv)} $\mathbf{G} \phi:= \neg \mathbf{F} \neg \phi \text{  (always)}$. 
The set $\mathcal{L}(\phi)$ comprises the language of finite-length words associated with the $LTL_F$ formula $\phi$. 
In this paper, we focus on a subset of $LTL_F$ called $safe-LTL_F$ \cite{saha2014automated}, that explicitly considers only safety properties \cite{kupferman1999model}. 

\begin{definition}[$safe-LTL_F$ Formula]\label{safeLTL}
An $LTL_F$ formula is a $safe-LTL_F$ formula if it can be written in \emph{positive normal form (PNF)}\footnote{In PNF, negations occur only adjacent to atomic propositions.}, using the temporal operators $\mathbf{X}$ (next) and $\mathbf{G}$ (always). 
\end{definition}

Next, we define an entity that will serve as an equivalent representation of an $LTL_F$ formula, and will allow us to check if the $LTL_F$ formula is satisfied or not. 

\begin{definition}[Deterministic Finite Automaton]\label{DFA}
A deterministic finite automaton (DFA) is a quintuple $\mathcal{A} = (Q, \Sigma, \delta, q_0,F)$ where $Q$ is a nonempty finite set of states, $\Sigma$ is a finite alphabet, $\delta : Q \times \Sigma \rightarrow Q$ is a transition function, $q_0 \in Q$ is the initial state, and $F \subseteq Q$ is a set of accepting states.
\end{definition}

\begin{definition}[Accepting Runs]\label{AccRun}
A \emph{run} of $\mathcal{A}$ of length $n$ is a finite sequence of $(n+1)$ states $q_0\xrightarrow{\sigma_0}q_1\xrightarrow{\sigma_1}\dots \xrightarrow{\sigma_{n-1}} q_n$ such that $q_{i} \in \delta (q_{i-1}, \sigma_{i-1})$ for all $i \in [1,n]$ and for some $\sigma_0,\dots,\sigma_{n-1} \in \Sigma$. 
The run is \emph{accepting} if $q_n \in F$. 
We write $\mathcal{L}(\mathcal{A})$ to denote the set of all words accepted by $\mathcal{A}$. 
\end{definition}
%

Every $LTL_F$ formula $\phi$ over $\mathcal{AP}$ can be represented by a DFA $\mathcal{A}_{\phi}$ with $\Sigma = 2^{\mathcal{AP}}$ that accepts all and only those runs that satisfy $\phi$, that is, $\mathcal{L}(\phi) = \mathcal{L}(\mathcal{A}_{\phi})$ \cite{de2015synthesis}. 
The DFA $\mathcal{A}_{\phi}$ can be constructed by using a tool like Rabinizer4 \cite{kvretinsky2018rabinizer}. 

\subsection{Discrete-time Dynamic Stochastic Games}

We model the interaction between the CPS (defender) and adversary as a two-player dynamic stochastic game that evolves according to some known (discrete-time) dynamics \cite{basar1999dynamic}. 
The evolution of the state of the game at each time step is affected by the actions of both players.

\begin{definition}[Discrete-time Dynamic Stochastic Game] 
A discrete-time dynamic stochastic game (DDSG) is a tuple $\mathcal{G} = (X, W, U_d, U_a, f, \mathcal{N}, \mathcal{AP}, L)$, where $X \subseteq \mathbb{R}^n$ and $W$ are Borel-measurable spaces representing the state-space and uncertainty space of the system, $U_d \subseteq \mathbb{R}^d$ and $U_a \subseteq \mathbb{R}^a$ are compact Borel spaces that denote the action sets of the defender and adversary, $f: X \times U_d \times U_a \times W \rightarrow X$ is a Borel-measurable transition function characterizing the evolution of the system, $\mathcal{N} = \{0,1,\dots,N-1\}$ is an index-set denoting the stage of the game, $\mathcal{AP}$ is a set of atomic propositions, and $L: X \rightarrow 2^{\mathcal{AP}}$ is a labeling function that maps states to a subset of atomic propositions that are satisfied in that state. 
\end{definition}

The evolution of the state of the system is given by: 
\begin{align}
x(k+1)&=f(x(k),u_d(k),u_a(k), w(k)); \quad x(0)=x_0 \in X; \quad k \in \mathcal{N}, \label{DDSGdynamics}
\end{align}
where $\{w(k)\}$ is a sequence of independent and identically distributed (i.i.d.) random variables with zero mean and bounded covariance.
%

%
%
%

In this paper, we focus on the \emph{Stackelberg setting} with the defender as leader and adversary as follower. 
The leader selects its inputs anticipating the worst-case response by the adversary. 
We assume that the adversary can choose its action based on the action of the defender \cite{ding2013stochastic}, and further, restrict our focus to stationary strategies for the two players. 
Due to the asymmetry in information available to the players, equilibrium strategies for the case when the game is zero-sum can be chosen to be deterministic strategies \cite{breton1988sequential}. 

\begin{definition}[Defender Strategy]
A stationary strategy for the defender is a sequence $\mu^{(d)}:=\{\mu_k^{(d)}\}_{k \in \mathcal{N}}$ of Borel-measurable maps $\mu_k^{(d)}: X \rightarrow U_d$. 
\end{definition}

\begin{definition}[Adversary Strategy]
A stationary strategy for the adversary is a sequence $\mu^{(a)}:=\{\mu_k^{(a)}\}_{k \in \mathcal{N}}$ of Borel-measurable maps $\mu_k^{(a)}: X \times U_d \rightarrow U_a$. 
\end{definition}
%
%
%

\section{Problem Formulation}\label{ProblemFormulation}

For a DDSG $\mathcal{G}$, recall that the labeling function $L$ indicates which atomic propositions are true in each state. 

\begin{assumption1}\label{labelfunc}
We restrict our attention to labeling functions of the form $L: X \rightarrow \mathcal{AP}$. 
Then, if $\mathcal{AP} = (a_1, \dots, a_p)$, $\mathcal{AP}$ and $L$ will partition the state space as $X:= \cup_{i=1}^p X_i$, where $X_i:=L^{-1}(a_i)$. 
We further assume that $X_i \neq \emptyset$ for all $i$. 
\end{assumption1}

\begin{remark}
Through the remainder of the paper, we interchangeably use $x_k$ or $x(k)$ to denote the state at time $k$. 
\end{remark}

Given a sequence of states $\mathbf{x}_{\mathcal{N}}:=(x_0, x_1,\dots, x_{N-1})$, using Assumption \ref{labelfunc}, if $\eta_k = L(x_k)$ for all $k \in \mathcal{N}$, then we can write $L(\mathbf{x}_{\mathcal{N}})=(\eta_0, \eta_1,\dots,\eta_{N-1})$. 

\begin{definition}[LTL Satisfaction by DDSG]\label{LTL+DDSG}
For a DDSG $\mathcal{G}$ and a $safe-LTL_F$ formula $\phi$, we write $\mathbb{P}_{\mu^{(d)}, \mu^{(a)}}^{x_0}\{L(\mathbf{x}_{\mathcal{N}}) \models \phi\}$ to denote the probability that the evolution of the DDSG starting from $x(0)=x_0$ under player policies $\mu^{(d)}$ and $\mu^{(a)}$ satisfies $\phi$ over the time horizon $\mathcal{N} = \{0,1,\dots,N-1\}$. 
\end{definition}

We are now ready to formally state the problem that this paper seeks to solve. 

\begin{problem}\label{Problem}
%
Given a discrete-time dynamic game $\mathcal{G} = (X, W, U_d, U_a, f, \mathcal{N}, \mathcal{AP}, L)$ that evolves according to the dynamics in Equation (\ref{DDSGdynamics}) and a $safe-LTL_F$ formula $\phi$, determine a policy for the defender, $\mu^{(d)}$, that maximizes the probability of satisfying $\phi$ over the time horizon $\mathcal{N} = \{0,1,\dots,N-1\}$ under any adversary policy $\mu^{(a)}$ for all $x_0 \in L^{-1}(a_j)$ for some $a_j \in \mathcal{AP}$. 
That is, compute:
\begin{align}
\sup_{\mu^{(d)}}~ \inf_{\mu^{(a)}}~ \mathbb{P}_{\mu^{(d)}, \mu^{(a)}}^{x_0}\{L(\mathbf{x}_{\mathcal{N}}) \models \phi\} \label{ProbSat}
\end{align}
\end{problem}

\section{Solution Approach}\label{SolutionApproach}

In this section, we present a dynamic programming approach to determine a solution to Problem \ref{Problem}. 
Our analysis is motivated by the treatment in \cite{ding2013stochastic} and \cite{yang2018dynamic}.

We then introduce the notion of secure control barrier certificates (S-CBCs), and use these to provide a lower bound on the probability of satisfaction of the $safe-LTL_F$ formula $\phi$ for a defender policy under any adversary policy in terms of the accepting runs of length less than or equal to the length of the time-horizon of interest of a DFA associated with $\phi$. 
For systems whose evolution of states can be written as a polynomial function of states and inputs, we present a sum-of-squares optimization approach in order to compute an S-CBC. 

S-CBCs generalize barrier certificates \cite{jagtap2018temporal} and control barrier certificates \cite{jagtap2019formal} to account for the presence of an adversary. 
A difference between the treatment in this paper and that of \cite{jagtap2018temporal}, \cite{jagtap2019formal} is that we define S-CBCs for stochastic dynamic games, while the latter papers focus on stochastic systems with a single agent.

\subsection{Dynamic Programming for $safe-LTL_F$ Satisfaction}

We introduce a dynamic programming (DP) operator that will allow us to recursively solve a Bellman equation related to Equation (\ref{ProbSat}) backward in time. 
First, observe that we can write the satisfaction probability in Definition \ref{LTL+DDSG} as:
\begin{align}
\mathbb{P}_{\mu^{(d)}, \mu^{(a)}}^{x_0}\{L(\mathbf{x}_{\mathcal{N}}) \models \phi\} &= \mathbb{E}_{\mu^{(d)}, \mu^{(a)}}\{\prod \limits_{k \in \mathcal{N}} \mathbf{1}(L(x_{k}) \models \phi)|x(0)=x_0\}, 
\end{align}
where $\mathbb{E}_{\mu^{(d)}, \mu^{(a)}}$ is the expectation operator under the probability measure $\mathbb{P}_{\mu^{(d)}, \mu^{(a)}}$ induced by agent policies $\mu^{(d)}$ and $\mu^{(a)}$. 
$\mathbf{1}(\cdot)$ is the indicator function, which takes value $1$ if its argument is true, and $0$ otherwise. 

Assume that $V : X \rightarrow [0,1]$ is a Borel-measurable function. 
A DP operator $T$ can then be characterized in the following way: 
\begin{align}
V(x_{N-1}) &=\mathbf{1}(L(x_{N-1}) \models \phi)\\
(TV)(x_k)&:= \sup_{u_d} ~ \inf_{u_a}~ \mathbf{1}(L(x_k) \models \phi) \int_X V(f(x_k,u_d,u_a,w))dx_{k+1}, \label{DPOperator}
\end{align}
where $dx_{k+1} \equiv (dx_{k+1}|x_k,u_d,u_a)$ is a probability measure on the Borel space $(X,\mathcal{B}(X))$. 

The following results adapts Theorem 1 of \cite{ding2013stochastic} to the case of temporal logic formula satisfaction over a finite time-horizon. 

\begin{theorem}\label{DPThm}
Assume that the DDSG $\mathcal{G}$ has to satisfy a $safe-LTL_F$ formula $\phi$ over horizon $\mathcal{N}$. 
Let the DP operator $T$ be defined as in Equation (\ref{DPOperator}). Additionally, if $dx_k \equiv (dx_{k+1}|x_k,u_d,u_a)$ is continuous, then, 
\begin{align}
\sup_{\mu^{(d)}}~ \inf_{\mu^{(a)}}~ \mathbb{P}_{\mu^{(d)}, \mu^{(a)}}^{x_0}\{L(\mathbf{x}_{\mathcal{N}}) \models \phi\} &= (T^NV) (x_0), 
\end{align}
where $T^N:=T \circ T \circ \dots \circ T$ ($N$ times) is the repeated composition of the operator $T$. 
\end{theorem}

\begin{proof}
Consider a particular pair of stationary agent policies $\mu^{(d)}$ and $\mu^{(a)}$. 
For these policies, define measurable functions $V_k^{\mu^{(d)}, \mu^{(a)}}: X \rightarrow [0,1]$, $k=0,1,\dots,N-1$: 
\begin{align}
V_{N-1}^{\mu^{(d)}, \mu^{(a)}}(x_{N-1}) &:= \mathbf{1}(L(x_{N-1})\models \phi)\\
V_k^{\mu^{(d)}, \mu^{(a)}}(x_k)&:=\mathbb{E}_{\mu^{(d)}, \mu^{(a)}}\{\prod \limits_{i=k} \limits^{N-1} \mathbf{1}(L(x_{i}) \models \phi)|x(k)=x_k \}, k=0,1,\dots,N-2 \label{CostToGo}
\end{align}
Therefore, we have $\mathbb{P}_{\mu^{(d)}, \mu^{(a)}}^{x_0}\{L(\mathbf{x}_{\mathcal{N}}) \models \phi\}  = V_0^{\mu^{(d)}, \mu^{(a)}}(x_0)$. 

Now, consider strategies of the agents at a stage $k$. Define the operator $T_{\mu^{(d)}_k, \mu^{(a)}_k}$: 
\begin{align}
(T_{\mu^{(d)}_k, \mu^{(a)}_k}V)(x_k)&:=  \mathbf{1}(L(x_k) \models \phi) \int_X V(f(x_k,u_d,u_a,w))dx_{k+1}
\end{align}

Expanding Equation (\ref{CostToGo}) using the definition of the expectation operator will allow us to write $V_k^{\mu^{(d)}, \mu^{(a)}}(x) = (T_{\mu^{(d)}_{k+1}, \mu^{(a)}_{k+1}}V)(x)$. 

The result follows by an induction argument which uses the fact that $T_{\mu^{(d)}_k, \mu^{(a)}_k}$ is a monotonic operator. 
We refer to \cite{ding2013stochastic} for details. 
Further, this procedure also guarantees the existence of a defender policy that will maximize the probability of satisfaction of $\phi$ under any adversary policy. 
\qed
\end{proof}

\subsection{Secure Control Barrier Certificates}

\begin{definition}\label{S-CBCDef}
%
A continuous function $B:X \rightarrow \mathbb{R}_{\geq 0}$ is a secure control barrier certificate (S-CBC) for the DDSG $\mathcal{G}$ if for any state $x \in X$ and some constant $c \geq 0$, 
\begin{align}
\inf \limits_{u_d}~ \sup \limits_{u_a} ~\mathbb{E}_w[B(f(x,u_d,u_a,w)|x] \leq B(x) +c.
\end{align} 
\end{definition}

Intuitively, for some defender action $u_d$, the increase in the value of an S-CBC is bounded from above along trajectories of $\mathcal{G}$ under any adversary action $u_a$. 

\begin{remark}
S-CBCs generalize control barrier certificates and barrier certificates seen in prior work. 
If $f(x,u_d,u_{a_1},w) \sim f(x,u_d,u_{a_2},w)$ for every $u_{a_1},u_{a_2} \in U_a$, then we recover the definition of a control barrier certificate \cite{jagtap2019formal}. 
The definition of a barrier certificate \cite{jagtap2018temporal}, \cite{prajna2007framework} is got by additionally requiring that $f(x,u_{d_1},u_{a_1},w) \sim f(x,u_{d_2},u_{a_2},w)$ for every $u_{d_1},u_{d_2} \in U_d$ and $u_{a_1},u_{a_2} \in U_a$. 
Here $\sim$ denotes stochastic equivalence of the respective stochastic processes \cite{pola2017bisimulation}. 
In the latter case, when $c=0$, the function $B$ is a \emph{super-martingale}. 
For this case, along with some additional assumptions on the system dynamics, asymptotic guarantees on the satisfaction of properties over the infinite time-horizon can be established \cite{prajna2007framework}. 
\end{remark}

\begin{remark}
Although our definition of S-CBCs in Definition \ref{S-CBCDef} bears resemblance to the notion of a \emph{worst-case barrier certificate} introduced in  \cite{prajna2007framework}, there are some distinctions. 
While the entity in \cite{prajna2007framework} considers a dynamical system with a single disturbance input, our setting considers three terms that influence the evolution of the state of the system: we want to find a defender input that will allow the barrier function to satisfy a certain property under any adversary input and disturbance. 
A second point of difference is that while \cite{prajna2007framework} focuses on asymptotic analysis, we consider properties over a finite time horizon. 
\end{remark}

We limit our attention to stationary strategies for both players. 
Studying the effects of other strategies is left as future work. 
The following preliminary result will be used subsequently to determine a bound on the probability of reaching a subset of states under particular agent policies over a finite time-horizon. 

\begin{lemma}\label{Kushner}
Consider a DDSG $\mathcal{G}$ and let $B:X \rightarrow \mathbb{R}_{\geq 0}$ be an S-CBC as in Definition \ref{S-CBCDef} with constant $c \geq 0$. 
Then, for any $\lambda > 0$ and initial state $x_0 \in X$, for a stationary defender policy, $\mu^{(d)} :X \rightarrow U_d$, the following holds under any stationary adversary policy $\mu^{(a)} :X \times U_d \rightarrow U_a$: 
\begin{align}
\inf \limits_{\mu^{(d)}}~ \sup \limits_{\mu^{(a)}} ~\mathbb{P}_{\mu^{(d)}, \mu^{(a)}}^{x_0}[\sup_{0 \leq k < N} B(x(k)) \geq \lambda] \leq \frac{B(x_0)+cN}{\lambda}
\end{align}
\end{lemma}
\begin{proof}
The proof follows from the result of Chapter III, Theorem 3 and Corollary 2-1 in \cite{kushner1967stochastic}, Definition \ref{S-CBCDef}, and the fact that the agents adopt stationary policies.
\qed
\end{proof}

\begin{definition}[$s-$Reachability]\label{safety}
For the DDSG $\mathcal{G}$ with dynamics in Equation (\ref{DDSGdynamics}), let $s \in [0,1]$ and $X_0 \subset X$ be the set of possible initial states and $X_1 \subset X$ be disjoint from $X_0$. Then, given $x_0 \in X_0$, $\mathcal{G}$ is $s-$reachable with respect to $X_1$, if $\sup\limits_{k \in \mathcal{N}}~\mathbb{P}^{x_0}[x_k \in X_1] \leq s$. 
That is, the probability of reaching a state in $X_1$ starting from $x_0 \in X_0$ in the time horizon $[0,N]$ is upper bounded by $s$. 
\end{definition}

\begin{theorem}\label{S-CBCThm}
With $X_0$ and $X_1$ known, and $X_0 \cap X_1 = \emptyset$, assume there exists an S-CBC $B: X \rightarrow \mathbb{R}_{\geq 0}$, stationary policies, $\mu^{(d)} :X \rightarrow U_d$ and $\mu^{(a)} :X \times U_d \rightarrow U_a$, and constant $c \geq 0$. Additionally, if there is a constant $\delta \in [0,1]$ such that: 
\begin{enumerate}
\item $B(x) \leq \delta$ for all $x \in X_0$,
\item $B(x) > 1$ for all $x \in X_1$, 
\end{enumerate}
then the DDSG $\mathcal{G}$ starting from $x_0 \in X_0$ is $(\delta+cN)-$reachable with respect to $X_1$. 
\end{theorem}
\begin{proof}
Observe that $X_1 \subseteq \{x \in X: B(x) \geq 1\}$. 
Therefore, starting from $x_0$, and following the respective agent policies, $\mathbb{P}_{\mu^{(d)}, \mu^{(a)}}^{x_0}[\exists k \in \mathcal{N}: x(k) \in X_1] \leq \mathbb{P}_{\mu^{(d)}, \mu^{(a)}}^{x_0}[B(x(k)) \geq 1]$. 
Since this should be true for arbitrary $k$, we have:
\begin{align*}
\sup_{k \in \mathcal{N}}\mathbb{P}^{x_0}[x_k \in X_1]& \leq \mathbb{P}_{\mu^{(d)}, \mu^{(a)}}^{x_0}\{\sup_{k \in \mathcal{N}} B(x(k)) \geq 1\} \leq \inf \limits_{\mu^{(d)}}~ \sup \limits_{\mu^{(a)}}~\mathbb{P}_{\mu^{(d)}, \mu^{(a)}}^{x_0}\{\sup_{k \in \mathcal{N}} B(x(k)) \geq 1\}\\
&\leq B(x_0)+cN \leq \delta+cN
\end{align*}
The second line of the above system of inequalities follows by setting $\lambda = 1$ in Lemma \ref{Kushner}, and the fact that $B(x) \leq \delta$ for all $x \in X_0$. 
\qed 
\end{proof}

\subsection{Automaton-Based Verification}

In order to verify that $\{L(\mathbf{x}_{\mathcal{N}}) \models \phi\}$ under agent policies $\mu^{(d)}$ and $\mu^{(a)}$, we need to establish that $(\eta_0, \eta_1,\dots,\eta_{N-1}) \subseteq \mathcal{L}(\mathcal{A}_{\phi})$. 
To do this, we first construct a DFA $\mathcal{A}_{\neg \phi}$, that accepts all and only those words over $\mathcal{AP}$ that do not satisfy the $safe-LTL_F$ formula $\phi$. 
We have the following result:
\begin{lemma} \cite{baier2008principles}
For $L(\mathbf{x}_{\mathcal{N}})=(\eta_0, \eta_1,\dots,\eta_{N-1})$ and a DFA $\mathcal{A}_{\phi}$, the following is true: 
\begin{align*}
(\eta_0, \eta_1,\dots,\eta_{N-1}) \subseteq \mathcal{L}(\mathcal{A}_{\phi}) \Leftrightarrow (\eta_0, \eta_1,\dots,\eta_{N-1}) \cap \mathcal{L}(\mathcal{A}_{\neg \phi}) = \emptyset
\end{align*} 
\end{lemma}

The construction of $\mathcal{A}_{\neg \phi}$ can also be carried out in Rabinizer4 \cite{kvretinsky2018rabinizer}. 
The accepting runs of $\mathcal{A}_{\neg \phi}$ of length less than or equal to $N$ can be computed using a depth-first search algorithm \cite{tarjan1972depth}. 
For the purposes of this section, it is important to understand that the accepting runs of $\mathcal{A}_{\neg \phi}$ of length less than or equal to $N$ will give a bound on the probability that a particular pair of agent policies $(\mu^{(d)},\mu^{(a)})$ will not satisfy $\phi$ over the time horizon $\mathcal{N}$. 
Using Definition \ref{AccRun} and following the treatment of \cite{jagtap2018temporal} and \cite{jagtap2019formal} define the following terms (the reader is also referred to these works for an example that offers a detailed treatment of the procedure): 
\begin{align}
\mathcal{R}_N(\mathcal{A}_{\neg \phi})&:=\{\mathbf{q} = (q_0,\dots,q_n) \in \mathcal{L}(\mathcal{A}_{\neg \phi}): n \leq N, q_i \neq q_{i+1} \forall i<n\}\label{AccRunNoLoop}\\
\mathcal{R}_N^a(\mathcal{A}_{\neg \phi})&:=\{\mathbf{q} = (q_0,\dots,q_n) \in \mathcal{R}_N(\mathcal{A}_{\neg \phi}): a \in \mathcal{AP} \text{ and } q_0 \xrightarrow{a} q_1\}\label{AccRunEachAP}\\
\mathcal{P}^a(\mathbf{q})&:=\begin{cases}\{(q_i,q_{i+1},q_{i+2}, T(\mathbf{q},q_{i+1})): 0 \leq i \leq n-2\}&\mathbf{q} \in \mathcal{R}_N^a(\mathcal{A}_{\neg \phi}), |\mathbf{q}|>2\\ \emptyset & otherwise \end{cases}\label{Len3RunHoriz}\\
T(\mathbf{q},q_{i+1})&:=\begin{cases} N+2-|\mathbf{q}|&\exists a \in \mathcal{AP}: q_{i+1} \xrightarrow{a} q_{i+1} \\1&otherwise \end{cases}\label{Horiz}
\end{align}

Intuitively, $\mathcal{R}_N(\mathcal{A}_{\neg \phi})$ is the set of accepting runs in $\mathcal{A}_{\neg \phi}$ of length not greater than $N$, and without counting any self-loops in the states of the DFA. 
The set $\mathcal{R}_N^a(\mathcal{A}_{\neg \phi})$ is the set of runs in $\mathcal{R}_N(\mathcal{A}_{\neg \phi})$ with the first state transition labeled by $a \in \mathcal{AP}$. 
For an element of $\mathcal{R}_N^a(\mathcal{A}_{\neg \phi})$, $\mathcal{P}^a(\mathbf{q})$ defines the set of paths of length $3$ augmented with a `loop-bound'. 
The `loop-bound' $T(\mathbf{q},q_{i+1})$ is an indicator of the number of `self-loops' the run in the DFA can make at state $q_{i+1}$ while still keeping its length less than or equal to $N$. 
We assume that $T(\mathbf{q},q_{i+1})=1$ when the run cannot make a self-loop at $q_{i+1}$. 
%
%

\subsection{Satisfaction probability using S-CBCs and $\mathcal{A}_{\neg \phi}$}

In this section, we show that an accepting run of $\mathcal{A}_{\neg \phi}$ of length less than or equal to $N$ gives a lower bound on the probability that a particular pair of agent policies will not satisfy the $safe-LTL_F$ formula $\phi$. 
We use this in conjunction with the S-CBC to derive an upper bound on the probability that $\phi$ will be satisfied for a particular choice of defender policy under any adversary policy. 
Specifically, we use Theorem \ref{S-CBCThm} over each accepting run of $\mathcal{A}_{\neg \phi}$ of length less than or equal to $N$ to give a bound on the overall satisfaction probability.

\begin{theorem}\label{SatProb}
Assume that the DDSG $\mathcal{G}$ has to satisfy a $safe-LTL_F$ formula $\phi$ over horizon $\mathcal{N}$. 
Let $\mathcal{A}_{\neg \phi}$ be the DFA corresponding to the negation of $\phi$, and for this DFA, assume that the quantities in Equations (\ref{AccRunNoLoop})-(\ref{Len3RunHoriz}) have been computed. 
Then, for some $a_j \in \mathcal{AP}$ and all $x_0 \in L^{-1}(a_j)$ the maximum value of the probability of satisfaction of $\phi$ for a defender policy $\mu^{(d)}$ under any adversary policy $\mu^{(a)}$ satisfies the following inequality:
\begin{align*}
\sup_{\mu^{(d)}}~ \inf_{\mu^{(a)}}~ 
\mathbb{P}_{\mu^{(d)}, \mu^{(a)}}^{x_0}\{L(\mathbf{x}_{\mathcal{N}}) \models \phi\} \geq 1-\sum_{\mathbf{q} \in \mathcal{R}_N^{a_j}(\mathcal{A}_{\neg \phi})}~\prod_{\rho \in \mathcal{P}^{a_j}(\mathbf{q})} ~(\delta_{\rho}+c_{\rho}T),
\end{align*}
where $\rho = (q,q',q'',T) \in \mathcal{P}^{a_j}(\mathbf{q})$ is the set of paths of length $3$ with loop bound $T$ for $a_j \in \mathcal{AP}$ in an accepting run of length $N$ in $\mathcal{A}_{\neg \phi}$. 
\end{theorem}
\begin{proof}
For $a_j \in \mathcal{AP}$, consider $\mathbf{q} \in \mathcal{R}_N^{a_j}(\mathcal{A}_{\neg \phi})$ (Equation (\ref{AccRunEachAP})) and the set $\mathcal{P}^{a_j}(\mathbf{q})$ (Equations (\ref{Len3RunHoriz}) and (\ref{Horiz})). 
Consider an element $\rho = (q,q',q'',T) \in \mathcal{P}^{a_j}(\mathbf{q})$. 
From Theorem \ref{S-CBCThm}, for some stationary defender policy $\mu^{(d)}$, the probability that a trajectory of $\mathcal{G}$ starting from $x_0 \in L^{-1}(\sigma: q \xrightarrow{\sigma} q')$ and reaching $x_1 \in L^{-1}(\sigma: q' \xrightarrow{\sigma} q'')$ under stationary adversary policy $\mu^{(a)}$ over the time horizon $T$ is at most $\delta_{\rho}+c_{\rho}T$. 
Therefore, the probability of an accepting run in $\mathcal{A}_{\neg \phi}$ of length at most $N$ starting from $x_0 \in L^{-1}(a_j)$ is upper bounded by:
\begin{align*}
\inf_{\mu^{(d)}}~ \sup_{\mu^{(a)}}~ 
 \mathbb{P}_{\mu^{(d)}, \mu^{(a)}}^{x_0}\{L(\mathbf{x}_{\mathcal{N}}) \models \neg \phi\} \leq \sum_{\mathbf{q} \in \mathcal{R}_N^{a_j}(\mathcal{A}_{\neg \phi})}~\prod_{\rho \in \mathcal{P}^{a_j}(\mathbf{q})} ~(\delta_{\rho}+c_{\rho}T)
\end{align*}

Now consider Equation (\ref{ProbSat}) of Problem \ref{Problem}. 
We have the following set of equivalences and inequalities: 
\begin{align*}
&\sup_{\mu^{(d)}}~ \inf_{\mu^{(a)}}~  \mathbb{P}_{\mu^{(d)}, \mu^{(a)}}^{x_0}\{L(\mathbf{x}_{\mathcal{N}}) \models \phi\} = 
\sup_{\mu^{(d)}}~(-\sup_{\mu^{(a)}}~(-\mathbb{P}_{\mu^{(d)}, \mu^{(a)}}^{x_0}\{L(\mathbf{x}_{\mathcal{N}}) \models \phi\}))\\
=&-\inf_{\mu^{(d)}}~\sup_{\mu^{(a)}}~(-\mathbb{P}_{\mu^{(d)}, \mu^{(a)}}^{x_0}\{L(\mathbf{x}_{\mathcal{N}}) \models \phi\})
=-\inf_{\mu^{(d)}}~\sup_{\mu^{(a)}}~(-1+\mathbb{P}_{\mu^{(d)}, \mu^{(a)}}^{x_0}\{L(\mathbf{x}_{\mathcal{N}}) \models \neg \phi\})\\
&\geq 1-\inf_{\mu^{(d)}}~\sup_{\mu^{(a)}}~\mathbb{P}_{\mu^{(d)}, \mu^{(a)}}^{x_0}\{L(\mathbf{x}_{\mathcal{N}}) \models \neg \phi\}
\geq 1-\sum_{\mathbf{q} \in \mathcal{R}_N^{a_j}(\mathcal{A}_{\neg \phi})}~\prod_{\rho \in \mathcal{P}^{a_j}(\mathbf{q})} ~(\delta_{\rho}+c_{\rho}T)
\end{align*}
%
%
%
\qed
\end{proof}
Theorem \ref{SatProb} generalizes Theorem 5.2 of \cite{jagtap2019formal} to provide a lower bound for a stationary defender policy that maximizes the probability that the $safe-LTL_F$ formula is satisfied by the DDSG $\mathcal{G}$ over the time horizon $\mathcal{N}$, starting from $x_0 \in L^{-1}(a_j)$ for some $a_j \in \mathcal{AP}$ for any stationary adversary policy. 
%
%
\subsection{Computing an S-CBC}

The use of barrier functions will circumvent the need to explicitly compute sets of reachable states, which is known to be undecidable for general dynamical systems \cite{lafferriere2001symbolic}. 
However, computationally efficient methods can be used to construct a barrier certificate if the system dynamics can be expressed as a polynomial \cite{prajna2007framework}. 
This will allow for determining bounds on the probability of satisfaction of the LTL formula without discretizing the state space. 
In contrast, if the underlying state space is continuous, computing the satisfaction probability and the corresponding agent policy using dynamic programming will necessitate a discretization of the state space in order to approximate the integral in Equation (\ref{DPOperator}).

We propose a sum-of-squares (SOS) optimization \cite{parrilo2003semidefinite} based approach that will allow us to compute an S-CBC  if the evolution of the state of the DDSG has a specific structure. 
The key insight is that if a function can be written as a sum of squares of different polynomials, then it is non-negative. 

\begin{assumption1}\label{SOS}
The sets $X, U_d, U_a$ in the DDSG $\mathcal{G}$ are continuous, and $f(x,u_d,u_a,w)$ in Equation (\ref{DDSGdynamics}) can be written as a polynomial in $x, u_d, u_a$ for any $w$. 
Further, the sets $X_i = L^{-1}(a_i)$ in Assumption \ref{labelfunc} can be represented by polynomial inequalities. 
\end{assumption1}

\begin{proposition}\label{CBCSSOS}
Under the conditions of Assumption \ref{SOS}, suppose that sets $X_0:=\{x \in X: g_0(x) \geq 0\}$, $X_1:=\{x \in X: g_1(x) \geq 0\}$, and $X:=\{x \in X: g(x) \geq 0\}$, where the inequalities are element-wise. 
Assume that there is an SOS polynomial $B(x)$, constants $\delta \in [0,1]$ and $c$, SOS (vector) polynomials $s_0(x), s_1(x)$, and $s(x)$, and polynomials $s^d_{u_i}(x)$ corresponding to the $i^{th}$ entry in $u_d$, such that: 
\begin{align}
&-B(x)-s_0^\intercal (x) g_0(x) + \delta \\
&B(x) -s_1^ \intercal (x)g_1(x) -1 \\
\forall u_a \in U_a: &-\mathbb{E}_w[B(f(x,u_d,u_a,w)|x] + B(x) -\sum_i (u_{d_i}-s^d_{u_i}(x)) - s^ \intercal (x) g(x) +c
\end{align}
are all SOS polynomials. Then, $B(x)$ satisfies the conditions of Theorem \ref{S-CBCThm}, and $u_{d_i}=s^d_{u_i}(x)$ is the corresponding defender policy. 
\end{proposition}

\begin{proof}
The proof of this result follows in a manner similar to Lemma 7 in \cite{wongpiromsarn2015automata} and Lemma 5.6 in \cite{jagtap2019formal}, and we do not present it here. \qed
\end{proof}

The authors of \cite{jagtap2019formal} discuss an alternative approach in the case when the input set has finite cardinality. 
A similar treatment is beyond the scope of the present paper, and will be an interesting future direction of research. 
%
%
%

\section{Example}\label{Example}

We present an example demonstrating our solution approach to Problem \ref{Problem}. 

\begin{example}
Let the dynamics of the DDSG $\mathcal{G}$ with $X=W=\mathbb{R}^2$, $U_d$ is a compact subset of $\mathbb{R}$, $U_a=[-1,1]$, and $w_1(k), w_2(k) \sim Unif[-1,1]$ (and i.i.d.) be given by: 
\begin{align}
x_1(k+1)&=-0.5x_1(k)x_2(k)+w_1(k)\\
x_2(k+1)&=x_1(k)x_2(k)+0.1x_2^2(k)+u_d(k)+0.6u_a(k)+w_2(k) \label{egEqn}
\end{align}
Let $\mathcal{AP} = \{a_0,a_1,a_2,a_3,a_4\}$, and sets $X_0, X_1, X_2, X_3,X_4$ such that for $x \in X_i$, $L(x) = a_i$. 
The sets $X_i$ are defined by: 
\begin{align*}
X_0&:=\{(x_1,x_2): x_1^2+x_2^2 \leq 0.9\},\\
X_1&:=\{(x_1,x_2): (2 \leq x_1 \leq 6) \wedge (-2 \leq x_2 \leq 2)\},\\
X_2&:=\{(x_1,x_2): x_1^2+(x_2-10)^2 \leq 4\},\\
X_3&:=\{(x_1,x_2): (-10 \leq x_1 \leq -3) \wedge (-4 \leq x_2 \leq -2)\},\\
X_4&:= X \setminus \bigcup_i X_i.
\end{align*}

The aim for an agent is to determine a sequence of inputs $\{u_d\}$ such that starting from $X_0$, for any sequence of adversary inputs $\{u_a\}$, it avoids obstacles in its environment, defined by the sets $X_1, X_2,$ and $X_3$ for $10$ units of time. 
The corresponding $safe-LTL_F$ formula is $\phi = [a_0 \wedge \mathbf{G} \neg (a_1 \vee a_2 \vee a_3)]$. 
The DFA that accepts $\neg \phi$ is shown in Figure \ref{DFA}. 
Suppose we are interested in determining a bound on the probability of $\phi$ being satisfied for a time-horizon of length $10$. Using Equations (\ref{AccRunNoLoop}) - (\ref{Horiz}), we have $\mathcal{P}^{a_0}(q_0, q_1, q_2) = \{(q_0,q_1,q_2,9)\}$, and $\mathcal{P}^{a_j} = \emptyset$ for $j=1,2,3,4$. 
\begin{figure}
\centering
\includegraphics[width = 3.4 in]{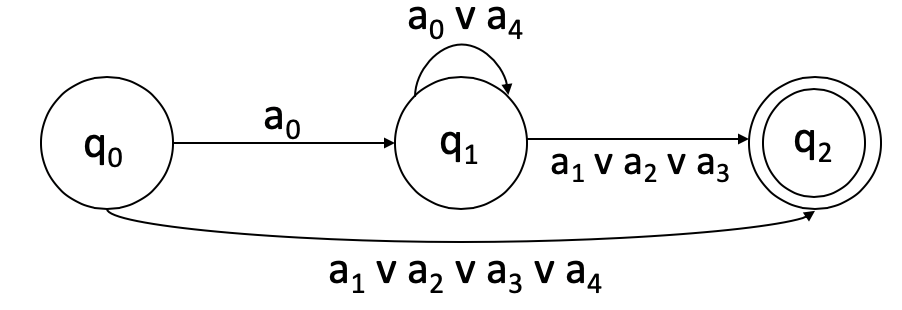}
\caption{\label{DFA}The DFA that accepts $\neg \phi$ for the $safe-LTL_F$ formula $\phi = [a_0 \wedge \mathbf{G} \neg (a_1 \vee a_2 \vee a_3)]$ and $\mathcal{AP} = \{a_0,a_1,a_2,a_3,a_4\}$.}
\end{figure}

We use a sum-of-squares optimization toolbox, SOSTOOLS \cite{prajna2002introducing} along with SDPT3 \cite{toh1999sdpt3}, a semidefinite program solver. 
The barrier function $B(x)=B(x_1,x_2)$ was assumed to be a polynomial of degree-two. 
For the case $c = 0$, we determine the smallest value of $\delta$ that will satisfy the conditions in Proposition \ref{CBCSSOS} to compute an S-CBC. 
The output of the program was an S-CBC given by
\begin{align*}
B(x)&=0.1915x_1^2 + 0.1868x_1x_2 -0.144x_1 + 0.1201x_2^2 + 0.1239x_2+0.16
\end{align*}

The environment and the obstacles denoted by the sets $X_1, X_2, X_3$ and the contours of the S-CBC is shown in Figure \ref{Regions}. 
We observe that $B(x)$ is less than $1$ in some part of $X_1$. 
A possible reason is that when solving for the second condition in Proposition \ref{CBCSSOS}, we work with the union of the sets $X_1, X_2,$ and $X_3$, which may lead to a conservative estimate of the S-CBC. 
\begin{figure}
\centering
\includegraphics[width = 4.9 in]{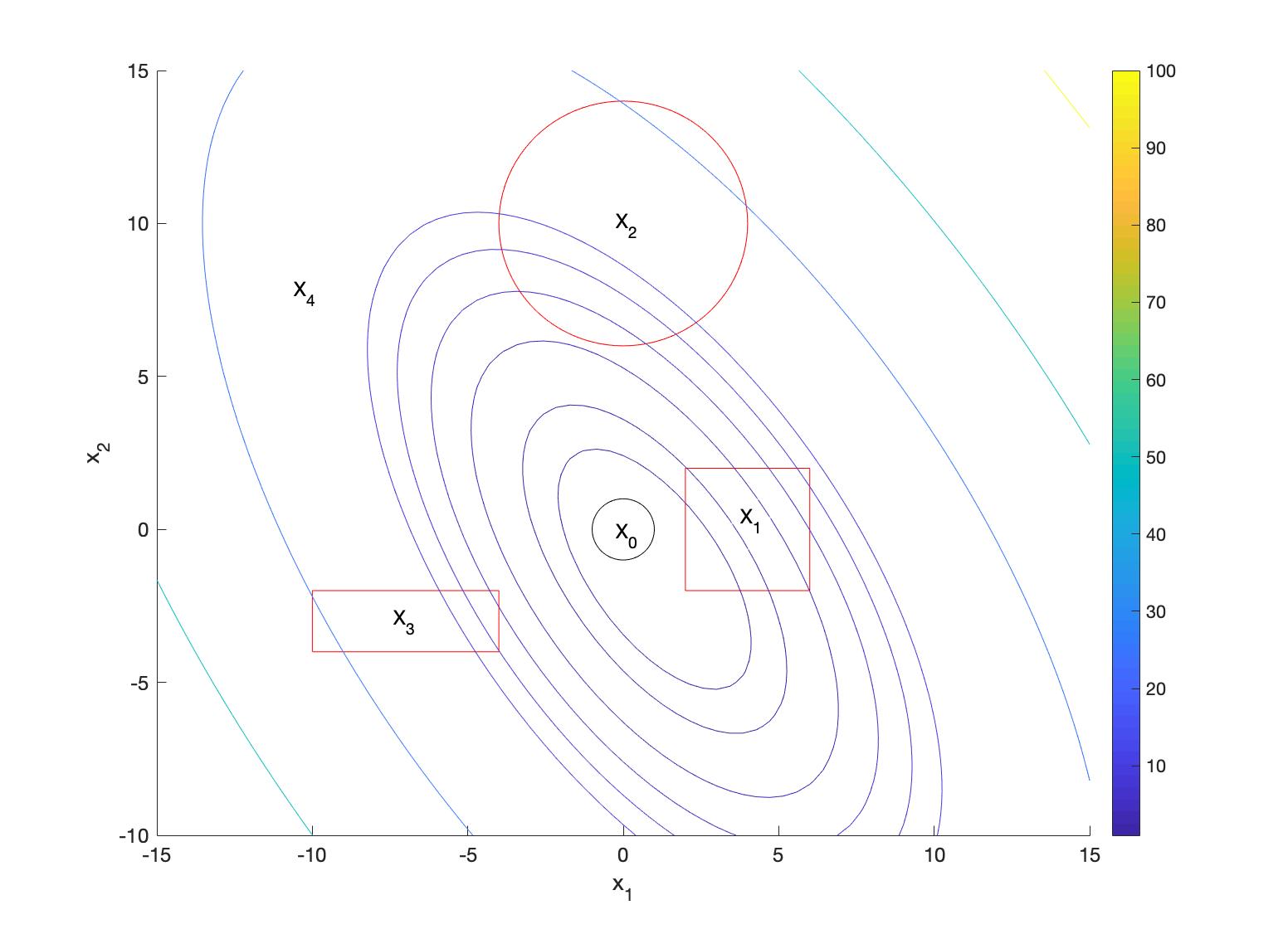}
\caption{\label{Regions}The regions $X_0, X_1, X_2, X_3,X_4$ along with the computed secure control barrier certificate (S-CBC): $B(x)=0.1915x_1^2 + 0.1868x_1x_2 - 0.144x_1 + 0.1201x_2^2 + 0.1239x_2+0.16$. The regions with red boundaries ($X_1, X_2, X_3$) denote obstacles in the environment. $X_0$ is the set from which the agent starts at time $0$. The contours show the values of the S-CBC of degree 2 ranging from $1$ to $100$.}
\end{figure}

From Theorem \ref{S-CBCThm} and the computed value of $\delta$, we have that
\begin{align*}
\sup \limits_{\mu^{(d)}}~ \inf \limits_{\mu^{(a)}}~  \mathbb{P}_{\mu^{(d)}, \mu^{(a)}}^{x_0}\{L(\mathbf{x}_{\mathcal{N}}) \models \phi\} \geq 0.9922.
\end{align*} 
This bound is conservative in the sense that we consider defender inputs $u_d$ for only the extreme values of $u_a = -1$ and $u_a =1$. 
However, for the dynamics in Equation (\ref{egEqn}), if the last inequality in Proposition \ref{CBCSSOS} is non-negative for both $u_a = -1$ and $u_a =1$, then for any $u_a \in [-1,1]$, this quantity will be non-negative. 

Determining methods to explicitly compute a defender policy and considering S-CBCs of higher degree is an area of future research. 

\end{example}

\section{Conclusion}\label{Conclusion}

This paper introduced a new class of barrier certificates to provide probabilistic guarantees on the satisfaction of temporal logic specifications for CPSs that may be affected by the actions of an intelligent adversary.
We 
presented a solution to the problem of maximizing the probability of satisfying a temporal logic specification in the presence of an adversary. 
The interaction between the CPS and adversary was modeled as a discrete-time dynamic stochastic game with the CPS as defender. The evolution of the state of the game was influenced jointly by the actions of both players. 
A dynamic programming based approach was used to synthesize a policy for the defender in order to maximize this satisfaction probability under any adversary policy. 
We introduced secure control barrier certificates, an entity that allowed us to determine a lower bound on the satisfaction probability. 
The S-CBC was explicitly computed for a certain class of dynamics using sum-of-squares optimization. 
An example illustrated our approach. 

Our example may have resulted in conservative bounds for the satisfaction probabilities since we restrict our focus to barrier certificates that are second degree polynomials and to stationary policies for the two agents. 
Future work will seek to study conditions under which possibly more effective non-stationary agent policies and higher degree S-CBCs can be deployed to solve the problem. 
A second interesting problem over a finite time-horizon is to investigate if explicit time bounds can be enforced on the temporal logic formula. 
An example of such a property is that the agent is required to reach a subset of states of the system between $3$ and $5$ minutes. 
This formula cannot be encoded in LTL, but there are other temporal logic frameworks like metric interval temporal logic \cite{alur1996benefits} or signal temporal logic \cite{maler2004monitoring} that will allow us to express it. 
We propose to study the case when the system will have to satisfy other kinds of timed temporal specifications \cite{bouyer2017timed} in the presence of an adversary in dynamic environments.

\bibliographystyle{splncs04}
\bibliography{ContBarrCert.bib}
\end{document}